\documentclass[12pt]{article}
\usepackage[a4paper]{geometry}
\usepackage{amssymb}
\usepackage[normalsize,normal,sc]{caption}

\newenvironment{blist}
        {
\begin{list}{$\bullet$}
        {
\setlength{\partopsep}{0.00in} 
\setlength{\topsep}{-0.02in} 
\setlength{\itemsep}{0.04in}
\setlength{\parsep}{0.00in}
\setlength{\leftmargin}{1cm}
        }
}
{\end{list}}

\newtheorem{thm}{Theorem}[section] 
\newtheorem{defi}{Definition}[section] 
\newtheorem{lem}{Lemma}[section]
\newtheorem{cor}{Corollary}[section] 
 
\newenvironment{proof}{\begin{trivlist} 
        \item[]\hspace{0cm}{\bf Proof.} 
\hspace{0cm} }{\hfill $\square$ 
                        \end{trivlist}}

\newenvironment{sketchproof}{\begin{trivlist} 
        \item[]\hspace{0cm}{\bf Sketchproof.} 
\hspace{0cm} }{\hfill $\square$ 
                        \end{trivlist}}

\def\bibfmta#1#2#3#4{{#1,} {#2}, {\em #3}, #4.}
\begin{document}

\title{\bf Reliable Self-Stabilizing Communication  
for Quasi Rendezvous}

\author{Colette JOHNEN $^a\quad$ Ivan LAVALL\'EE $^b\quad$  
Christian LAVAULT $^c\,$\thanks{Corresponding author:  
LIPN, CNRS ESA 7030, Universit\'e Paris-Nord, 
99, av. J.-B. Cl\'ement 93430 Villetaneuse, France. Email: 
lavault@lipn.univ-paris13.fr} \\
$^a$\ {\small \sl LRI-CNRS Universit\'{e} Paris-Sud} \\
$^b$\ {\small \sl LRIA, Universit\'{e} Paris 8} \\
$^c$\ {\small \sl LIPN-CNRS, Universit\'e Paris-Nord}
}
\date{\empty}
\maketitle
\vspace{-.5cm}

\begin{abstract}
The paper presents three self-stabilizing protocols for basic 
fair and reliable link communication primitives. We assume a 
link-register communication model under read/write atomicity, 
where every process can read from but cannot write into its 
neighbours' registers. The first primitive guarantees that 
any process writes a new value in its register(s) only after 
all its neighbours have read the previous value, whatever the 
initial scheduling of processes' actions. The second primitive 
implements a ``weak rendezvous'' communication mechanism by 
using an alternating bit protocol: whenever a process consecutively 
writes $n$ values (possibly the same ones) in a register, each 
neighbour is guaranteed to read each value from the register 
at least once. On the basis of the previous protocol, the 
third primitive implements a ``quasi rendezvous'': in words, 
this primitive ensures furthermore that there exists exactly 
one reading between two writing operations

All protocols are self-stabilizing and run in asynchronous 
arbitrary networks. The goal of the paper is in handling 
each primitive by a separate procedure, which can be used 
as a ``black box'' in more involved self-stabilizing protocols.

\end{abstract}

\newpage
\section{Introduction}
A {\em self-stabilizing} system which is started from an arbitrary 
initial configuration, regains its consistency and demonstrates 
legal behaviour by itself, without any outside intervention. 
Consequently, a self-stabilizing system needs not be initiated 
to any configuration, and can recover from {\em transient faults}. 
More precisely, it can recover from {\em memory corruptions} 
and copes with processors or channels crashes and recoverings 
(i.e., dynamic networks).

\subsection{The Communication primitives}
In the paper, we present fair and reliable self-stabilizing 
communication primitives in the link-register model. The communication 
between two neighbours ($A$ and $B$) is carried out by the use of 
two sets of {\em communication registers} called {\em registers}: 
$r_{AB}$ and $r_{BA}$. Process $A$ can write in the registers of 
$r_{AB}$ and each process $A$ and $B$ can read from the registers 
of $r_{AB}$.The registers support {\em read} and {\em write} atomic 
operations. For example, let $\Sigma=\{a,b,c,\epsilon\}$ be an alphabet 
and $w=aaabbbbcc=a^3b^4c^2$ a sequence of valuewritten by $A$ into 
$r_{AB}$. The communication primitives in their very first basic 
form {\em do not ensure more} than e.g.: $a^*b^*c^*$ is eventually 
read by $B$. 

\medskip 
The first presented primitive guarantees that any process $A$ writes 
a new value in its register(s) $Write_{AB}$ only after its neighbour 
$B$ has read the previous value. Notice that when $A$ writes $n$ 
times the same value consecutively in the register $Write_{AB}$, 
the primitive ensures that $B$ eventually copies this value at 
least once. For example, given $\Sigma$ and $w$ as above, the 
first primitive {\em only guarantees} that e.g., $aa^*bb^*cc^*$ 
is eventually read by each neighbour: each symbol in $w$, ($a$, $b$ 
and $c$) is read at least once, whatever the number of occurrences. 
This primitive simulates self-stabilizing reliable message-passing 
communication in the link-register asynchronous model. It guarantees 
that a message, that is the value of the register $Write$, is 
eventually received: the value is eventually known from the 
neighbours' process.

\medskip
The {\em rendezvous} mechanism (as defined in~\cite{hoa}) 
synchronizes communications, i.e., the $write$ and $read$ 
operations are performed in and from the same register. 
When Process $A$ writes a value in its register $Write_{AB}$, 
it cannot perform any other action until process $B$ has 
completed a $read$ operation from the register $Write_{AB}$.

The second communication primitive is a self-stabilizing 
``weak rendezvous''. After performing a $write$ operation in 
its register $Write_{AB}$, the process $A$ cannot perform but 
some specific actions, as long as process $B$ has not completed 
a $read$ operation from $Write_{AB}$. Therefore, if $A$ 
consecutively writes $n$ values (possibly the same ones) in 
the register $Write_{AB}$, the primitive guarantees that $B$ 
eventually copies each value at least once. If $A$ writes $n$ 
times the same value in $Write_{AB}$, the value will be read 
at least $n$ times. As an example, given $\Sigma$ and $w$ as 
above, the second primitive {\em at least guarantees} that e.g., 
$a^3a^*b^4b^*c^2c^*$ is eventually read by each neighbour: 
each symbol in $w$ ($a$, $b$ and $c$) is read at least the number 
of times the symbol occurs in $w$ (but any symbol may be read 
strictly more than its number of occurrences).

The third self-stabilizing communication primitive performs 
a quasi synchronization. It is a ``quasi rendezvous'' mechanism 
and requires that between two $write$ operations performed by 
the process $A$ in $Write_{AB}$, the process $B$ cannot perform 
but one and only one {\em read} operation from $Write_{AB}$. 
Therefore, if $A$ writes $n$ consecutive times the same value 
(possibly the same one in each row) in the register $Write_{AB}$, 
the primitive guarantees that $B$ will copie each of the $n$ values 
exactly one time, once the system is stabilized. For example, 
given again $\Sigma$ and $w$ as above, the third primitive 
{\em does ensures} that {\em exactly} $a^3b^4c^2$ is eventually 
read by each neighbour: each symbol in $w$ ($a$, $b$ and $c$) 
is read exactly the number of times it occurs in $w$.

\medskip
Each such primitive may prove useful as a communication ``black box'' 
in designing more involved distributed self-stabilizing protocols.

\subsection{Related Works and Results} \label{related}
A deterministic self-stabilizing ``balance-unbalance'' mechanism 
on two processes systems under read/write atomicity is presented 
in~\cite{dim1} and in~\cite{dim2}. The two processes are not 
executing the same code. The one executes the balance code: 
when both processes have the same color, it changes color. 
The other executes the unbalance code: when both processes 
have not the same color, it changes color.
In~\cite{dim1}, this mechanism is used to guarantee that each 
process has a mutual exclusion access to a critical section, 
and in~\cite{dim2}, it is used to ensure synchronization 
of the processes. In both cases, this mechanism provides strong 
synchronization: between two ``actions'' of a process, the other 
process cannot perform but only one ``action''.
In~\cite{dim1,dim2}, the two processes protocol is used to design
a mutual exclusion algorithm (global synchronization) on tree 
networks. As claimed in~\cite{dim1,dim2}, the balance-unbalance 
mechanism cannot be extended to any network topology, since there 
exist no deterministic self-stabilizing synchronization protocols 
in uniform arbitrary networks. On the other hand, a self-stabilizing 
synchronization on unidirectional rings is provided in~\cite{bgj} 
through the deterministic token circulation mechanism: 
between two actions of a process its neighbours cannot perform 
but only one action.

Any self-stabilizing reset protocol~\cite{AG94,AKY90,APV91}
can be combined with the protocol in~\cite{A85} to design 
a self-stabilizing synchronizer. General self-stabilizing 
synchronizers are presented e.g. in~\cite{AV91,AKMPV93,V94}. 
Global self-stabilizing synchronizers for tree
networks are also proposed in~\cite{dim2,ABDT98,BDPV98}.
A self-stabilizing local synchronizer, that synchronizes 
each node in a tree network with its neighbours is presented 
in~\cite{JADT99}.
In the recent literature, several communication problems in the 
message-passing model have been addressed. A self-stabilizing 
communication protocol for two-way handshake is presented 
in~\cite{gom}, and a self-stabilizing version of the alternating-bit 
protocol is given in~\cite{afb}. 
In~\cite{anh}, Anagnostou and Hadzilacos present a self-stabilizing 
data link protocol under the read/write atomicity model such that, 
between two $write$ operations in the register, only one $read$ 
operationfrom that register is performed. However, no proof of the 
protocol is given in their paper. By contrast, our last two primitives 
use the alternating-bit mechanism, and since the two bits values must 
begin with the same value 0, our algorithm in section~\ref{qrv} is 
twice as fast as in~\cite{anh}.

\medskip
Section~\ref{model} describes our model with the basic assumptions. 
In Section~\ref{princ}, we present the general principle of our 
solution for a two processes system. The generalization to $n$ 
processes in arbitrary networks yields the {\em Read Checking} 
self-stabilizing protocol, which is presented in Section~\ref{proto}. 
Section~\ref{proper} is devoted to the proof of liveness 
and correctness of the Read Checking protocol. Section~\ref{wrv} 
presents the weak rendezvous protocol and Section~\ref{qrv} describes 
our quasi rendezvous protocol. Finally, the paper ends with few 
concluding remarks.

\section{Model and Requirements} \label{model}
Although distinct from the one described in~\cite{dim1}, 
our model relies on close requirements and assumptions, 
especially in terms of communication (e.g., link registers, 
read/write atomicity, etc.). A distributed system consists 
of $n$ processes denoted $A$, $B$, etc. Each process resides 
on a node of the system's {\em communication graph } 
(or {\em network }). Two processes which reside on two 
adjacent nodes of the network are called {\em neighbours}. 
We model distributed self-stabilizing systems as a set of 
(possibly infinite) state machines called processes. 
Each process can only communicate with the subset of processes 
consisting of its neighbours. We assume a {\em link-register} 
communication model under read/write atomicity~\cite{dim1}. 
Each link between any two neighbours $A$ and $B$ is composed 
of two pairs of registers\footnote{In our model, the registers 
are physical (hardware) devices. Reading from or writing 
in {\em one} register is an atomic action according to the 
design of the  microprocessor.}, denoted $(Write_{AB},Read_{AB})$ 
and $(Write_{BA},Read_{BA})$, and belonging to $A$ and $B$, 
respectively. Process $A$ can read from the two registers 
of $B$, $Write_{BA}$ and $Read_{BA}$, but cannot write into 
them. Similarly, process $A$ cannot write but in its own 
registers, $Write_{AB}$ and $Read_{AB}$, to communicate 
with $B$.

A {\em configuration} of the system is the vector of states 
of all processes. The state of a process is the value of its 
internal variables and the contents of its registers.

\subsection{Schedulers, Demons and Computation}
An {\em atomic step} is the ``largest'' step which is guaranteed 
to be executed uninterruptedly. A process uses {\em read/write} 
atomicity if each atomic step contains either a single read 
operation or a single write operation but not both. The system 
behaviour is modelled by the interleaving model in which processes 
are activated by a scheduler. The scheduler is regarded as a 
{\em fair} adversary: in a self-stabilizing system, all possible 
fair executions are required to converge to a correct behaviour. 
A fair scheduler shall eventually activate any process which may 
continuously perform an action. A common scheduler activates 
either processes one by one (central demon) or subsets of processes 
(distributed demon). Under read/write atomicity, both central and 
distributed schedulers/demons are ``equivalent'', in the sense that 
any execution performed under a distributed scheduler may be 
simulated by a central one. A process 
which can perform an atomic step into a configuration $c$, 
is said to be {\em enabled} at $c$. During a {\em computation step}, 
one or more processes execute an atomic step. A {\em computation} 
of a protocol ${\cal P}$ is a sequence of configurations 
$c_1,c_2,\ldots$ such that, for $i=1, 2, \ldots$, the configuration 
$c_{i+1}$ is reached from $c_i$ by one computation step. 
A computation is said to be {\em maximal} either if the sequence 
is infinite, or if it is finite and no process is enabled in the 
final configuration. A {\em problem} is a predicate defined on 
computations.

\subsection{Self-Stabilization}
The protocol ${\cal P}$ is {\em self-stabilizing} for the 
problem $\Pi$ if and only if there exists a predicate ${\cal L}$ 
defined on configurations such that:
\begin{blist}{}
\item 
all computations reach a configuration that satisfies $\cal L$ 
({\bf convergence}); 
\item  
all computations, from $\cal L$, satisfy problem $\Pi$ 
({\bf correctness}). 
\end{blist}

\medskip
Notice that the maximal computations of a self-stabilizing 
protocol may be finite; in that case the algorithm is said 
to be {\em silent}~\cite{dgs}. Most self-stabilizing algorithms 
which build spanning tree or elect a leader are silent~\cite{hc}.
Self-stabilizing protocols offers full and automatic 
protection against all transient process failures, no matter 
how much the data have been corrupted: e.g., all registers 
values may be fully corrupted. 
 
So, whatever the registers values, our protocols secure the 
transfer of information between any two pair of neighbours 
after a ``certain delay time''. 

\section{Principle of the Solution} \label{princ} 
Let a two processes system, consisting in two neighbouring
processes $A$ and $B$ equipped with their two pairs of registers 
(see Section~\ref{model}). The principle of the solution for 
$A$ relies on the following basic idea. Under read/write atomicity, 
$A$ systematically keeps reading the value from $Write_{BA}$ and 
copies out this value in $Read_{AB}$ (i.e., $A$ reads the message 
sent by $B$ and copies out the message in $Read_{AB}$ to inform 
$B$ that its message is received). Besides, $A$ systematically 
keeps reading the value from $Read_{BA}$ and compares it to 
the value of $Write_{AB}$. When both values are equal, $A$ 
finds out that $B$ somehow read that value (i.e., the information 
has been transmitted), So it can stop reading and can write 
again in $Write_{AB}$.

\noindent \hrulefill\-\hrulefill

\noindent {\bf while} {\em true} {\bf do} \\
\hspace*{1cm} $A$ writes in $Write_{AB}$ \\
\hspace*{.5cm} {\bf repeat} \\
\hspace*{1cm} $A$ reads from $Write_{BA}$ ; \\
\hspace*{1cm} $A$ writes out the value of $Write_{BA}$ into 
$Read_{AB}$ ; \\
\hspace*{1cm} $A$ reads from $Read_{BA}$ \\
\hspace*{.5cm} {\bf until}\ \, $Read_{BA}=Write_{AB}$ \\
{\bf endwhile}

\smallskip
\begin{center}
{\bf Fig.~1.}\ {\em The basic 2-processes protocol for $A$.}
\end{center}
\vspace{-3mm}

\noindent \hrulefill\-\hrulefill

\bigskip
After $A$ has written a new value in $Write_{AB}$, 
$A$ becomes ``weakly locked'' until $B$ receives the message 
($Read_{BA}=Write_{AB}$). When $A$ is inside the {\bf repeat} 
loop, it can only perform some actions, for instance, $A$ 
cannot write in its register $Write_{AB}$. 

\smallskip
In a self-stabilizing setting, $A$ may then proceed with the
execution of its own code, since the protocol makes it sure that 
$B$ did read the value from $Write_{AB}$ (at least, it results 
from the protocol that $A$ knows for sure that the values in 
$Read_{BA}$ and $Write_{AB}$ are identical). The corresponding 
code sequence for $B$ is of course fully symmetrical to the 
basic protocol for $A$: the roles of $A$ and $B$ (i.e. the 
registers' names) have simply to be inverted within the above 
protocol in Fig.~1. Thus, a two-way communication is established 
between $A$ and $B$.
 
\section{The Protocol in Arbitrary Networks} \label{proto}
The generalization of the above protocol to a system of $n>2$ 
processes constituting an arbitrary network is now easy. We 
still assume each pair of neighbouring processes in the network 
to be equipped with its two pairs of registers on their common link. 
In order to simplify the use of variables, we call {\em ``message''} 
the ``information'' exchanged between neighbours during the 
execution of the protocol.

A protocol which stabilizes on a single link may not generalize 
to a protocol which stabilizes on all links of a (finite) network, 
e.g. by having each process execute the ``link-protocol'' 
in a round robin manner on each individual link adjacent to it. 
Taking the $n$-processes system pair by pair may cause a deadlock: 
for all $i\in \{0,\ldots,n-1\}$, $A_i$ may be waiting for 
$A_{i+1}$ to read from $Write_{A_iA_{i+1}}$, with $A_n=A_0$.

\subsection{Notation} \label{nota}
\hfill

{\bf Write register for $\mathbf{A}$:} $Read_{AB_i}$ is the register 
in which $A$ writes the value of the last message read by $A$ 
and sent by $B_i$.

\smallskip
{\bf Read register for $\mathbf{A}$:} $Write_{B_iA}$ is the register 
in which $B_i$ writes the message to be transmitted to $A$, and 
$Read_{B_iA}$ is the register in which $B_i$ writes the value of 
the last message read by $B_i$ and sent by $A$.

\smallskip
{\bf Write and read register for $\mathbf{A}$:} $Write_{AB_i}$ 
is the register in which $A$ writes the value of the message 
which is to be sent to its $i\/$th neighbour $B_i$.
 
\smallskip
{\bf Function {\em get}$_{\mathbf{i}}$ for $\mathbf{A}$:} 
{\em get}$_i$ takes no argument and returns the next message 
to be sent to the $i\/$th neighbour of $A$ ({\em get}$_i$ is 
a helper function added to $A$). 

\subsection{The Read Checking Protocol} \label{rcproto}
On the same assumptions for the model (read/write atomicity) and 
for the scheduler's actions (rules of activations of processes 
and fairness) as given in Section~\ref{model}, the specification 
of the self-stabilizing Read Checking protocol in arbitrary 
networks for a process $A$, with neighbours $B_i$'s 
$(1\le i\le N_A)$, is as follows.

%\newpage
\noindent \hrulefill\-\hrulefill

\smallskip
\noindent {\bf constant}\ \, $N_A$\hspace{.8cm} : the number of 
neighbours of $A$ ; \\
\noindent {\bf var}\ \, $s_i$\hspace*{1.9cm} : message to be
sent to the $i\/$th neighbour of $A$ ; \\
\hspace*{.7cm} $r_i$\hspace*{2cm} : message sent from the $i\/$th
neighbour of $A$ ; \\
\hspace*{.7cm} $val_i$\hspace*{1.7cm} : value of the last message
sent from $A$ and read by the $i\/$th neighbour of $A$ ;

\medskip
\noindent {\bf while} {\em true} {\bf do} \\
\hspace*{.5cm} {\bf for} $\;i=1\;$ {\bf to} $\;N_A\;$ {\bf do} \\
\hspace*{2cm} {\bf write}$(Write_{AB_i},get_i)$ ; \\
\hspace*{.5cm} {\bf endfor} \\
\hspace*{.5cm} {\bf repeat} \\
\hspace*{1cm} {\bf for} $\;i=1\;$ {\bf to} $\;N_A\;$ {\bf do} \\
\hspace*{2cm} $r_i\gets$ {\bf read}$(Write_{B_iA})$ ; \\
\hspace*{2cm} {\bf write}$(Read_{AB_i},r_i)$ ; \\
\hspace*{2cm} $val_i\gets$ {\bf read}$(Read_{B_iA})$ ; \\
\hspace*{2cm} $s_i\gets$ {\bf read}$(Write_{AB_i})$ ; \\
\hspace*{1cm} {\bf endfor} \\
\hspace*{.5cm} {\bf until}\ \, \,( $\forall i\in [1,N_A]\ \; 
val_i = s_i$ ) \\
{\bf endwhile}

\smallskip
\begin{center}
{\bf Fig.~2.}\ {\em The Read Checking protocol for $A$.}
\end{center}
\vspace{-.3cm}
\noindent \hrulefill\-\hrulefill

\section{Proof of the Read Checking Protocol} \label{proper}

\subsection{Proof of Liveness} \label{live}
%\begin{lem} \label{lock}
%Let $\gamma$ be any configuration of an arbitrary network of 
%processes on which the read checking protocol is performed. 
%All processes are enabled in configuration $\gamma$.
%\end{lem}
%\begin{proof} Let $A$ be a process, its program counter is such that 

%$\bullet\; A$ is not in the {\bf repeat} loop, and hence $A$ can write 
%into one of its {\em Write} registers;

%$\bullet\; A$ is in the {\bf repeat} loop, and hence $A$ can either 
%read from one of its neighbours' register, or write into one of its 
%{\em Read} registers. Thus, in all configuration, $A$ can perform 
%an atomic step (if chosen by the scheduler).
%\end{proof}

%\begin{lem} \label{inf}
%Every execution of the protocol on any arbitrary network is infinite.
%\end{lem}
%\begin{proof} From Lemma~\ref{lock}, whatever the current configuration, 
%all processes can execute an action. Hence, every configuration is 
%deadlock-free and no execution can reach a deadlock configuration. 
%Therefore, every execution is infinite.
%\end{proof}

\begin{lem}
Whatever the execution, every process performs an infinite number 
of actions.
\end{lem}
\begin{proof} Read/write atomicity ensure that each process is always 
enabled. Therefore, every execution is infinite (every configuration 
is deadlock-free), and in each configuration that is reached every 
process can perform an action (fair scheduler). The scheduling of 
processes' actions is fair: if a process can always execute an action, 
then the process finally performs an action. Thus, by fairness, every 
process is performing an infinite number of actions, whatever the 
execution.
\end{proof}

\begin{lem} \label{allow}
Let A be a process with its program counter in the {\bf repeat} loop 
and let B be a neighbour of A. Whatever the current configuration 
and the execution, the processes system executing the protocol either 
eventually reaches a configuration in which B allows A to write, 
or A exits the {\bf repeat} loop.
\end{lem}
\begin{proof}
Suppose $B$ never allows $A$ to write and $A$ never exits the 
{\bf repeat} loop. Then $A$ never changes the value in its register 
$Write_{AB}$. Under these conditions, updating its register 
$Read_{BA}$ is a writing permission given to $A$ by $B$ 
(since between the reading of the value from the register 
$Write_{AB}$ and the writing of that value in $Read_{BA}$, 
the register $Write_{AB}$ does not change value).

Whatever the current configuration and the execution, if the program 
counter of $B$ is not within the {\bf repeat} loop, it takes $B$ less 
than $N_{B}$ actions to enter the {\bf repeat} loop. Once $B$ enters 
the loop, after $4N_{B}$ actions, it updates all its {\em Read} 
registers, and thus allows  $A$ to write.

Whatever the current configuration and the execution, if the program 
counter of $B$ is within the {\bf repeat} loop, it takes $B$ at least
$4N_{B}$ actions either to exit the loop, or to update its register 
{\em Read}$_{AB}$.

Whatever the execution, $B$ performs an infinite number of actions 
(by Lemma~\ref{allow}) and eventually, either $B$ allows $A$ 
to write, or $A$ exits the {\bf repeat} loop.
\end{proof}
\begin{defi} \label{def1}
Let A and B be two neighbouring processes. A is said to allow B 
to write iff Read$_{BA}=$ Write$_{AB}$. Let A be a process 
and let $N_A$ denote the number of neighbours of A ($N_A$ is 
the degree of A in the network).
\end{defi}

\begin{defi}
Let A and B be two neighbouring processes. The update of the register 
Read$_{AB}$ is the sequence of the two following actions performed 
by B: $r_i\gets$~{\bf read}$(Write_{AB})$~; {\bf write}$(Read_{BA},r_i)$.

A {\em wrong} writing is a write action in the register Read$_{BA}$ 
which is not performed within the context of an update.
(The {\em correct} writing into the register Read$_{BA}$ is a write 
action executed within the context of an update.)
\end{defi}

\begin{lem} \label{wrong}
After executing its first action, no process can perform a wrong writing.
\end{lem}
\begin{proof} Process $A$ can perform at most one {\em wrong} writing, 
and it may only happen when initially its program counter is set up after 
reading from the {\em Write} register and before writing in the {\em Read}
register. Once this write action is executed, each write action of $A$ 
in a {\em Read} register is performed within the context of an update.
\end{proof}

\begin{lem} \label{stop}
Let A and B be two neighbouring processes. After B executes its first 
action, if B allows A to write, then only the writing of A in its
register Write$_{AB}$ may be able to cancel that permission.
\end{lem}
\begin{proof} Nothing but writing into the register $Read_{BA}$ 
or into the register $Write_{AB}$ can cancel the writing permission. 
After $B$ executes its first action, from Lemma~\ref{wrong} there is 
no {\em wrong} writing anymore. Hence, any writing into the register 
$Read_{BA}$ is executed within the context of a register's update. 
This update is such that the permission remains given to $A$, unless 
$A$ writes into its register $Read_{BA}$ during the updating 
process or after the last update.
\end{proof}

\begin{thm} \label{exit}
Let A be a process. Whatever the execution, the system of processes 
which performs the protocol reaches a configuration in which 
A is not within the {\em \bf repeat} loop anymore.
\end{thm}
\begin{proof} Suppose $A$ remains within the {\bf repeat} loop 
forever; then $A$ never writes into its {\em Write} registers.
Every $4N _A$ actions, $A$ is checking out the loop exiting 
condition. Whatever the execution, process $A$ performs 
an infinite number of actions. Hence, $A$ checks out the 
{\bf repeat} loop exiting condition an infinite number of times. 
In particular, $A$ tests the exit condition an infinite number
of times after all its neighbours have already executed an
action.

If at some test all neighbours of $A$ allow its writing, 
then, at the next test, all its neighbours keep on giving $A$ 
permission to write (by Lemma~\ref{stop}). In the meanwhile, 
$A$ has updated its variables $r_i$ and $s_i$, and when the test 
happens, the loop exiting condition is satisfied: $A$ exits 
the loop.

Process $A$ stays within the loop infinitely long in the case 
when, at each test, at least one neighbour does not allow its 
writing. Once a neighbour has allowed $A$ to write, this neighbour 
cannot withdraw permission from $A$. Therefore, there exists at 
least one neighbour of $A$ which never allows $A$ to write. 
Now from Lemma~\ref{allow}, this is impossible, and the theorem 
follows. Therefore, the protocol is deadlock-free.
\end{proof}

\begin{cor}
Let A be a process. Whatever the execution, A writes an infinite 
number of times into all its Write registers.
\end{cor}
\begin{proof} If $A$ is out of the loop, then it takes $A$ less than  
$N _A$ actions to enter the loop. When it is within the {\bf repeat} 
loop, then by Theorem~\ref{exit}, $A$ cannot stay infinitely 
long. $N_A$ actions after exiting the loop, $A$ writes into all 
its {\em Write} registers and reenters the {\bf repeat} loop.
\end{proof}

\subsection{Correctness Proof of the Read Checking Protocol}\label{correct} 
\begin{thm} \label{proof}
Let A and B be two neighbouring processes. After B executes 
its first action and after any writing in the register 
Write$_{AB}$, A can write in the register Write$_{AB}$ 
only if B allows it, i.e. Read$_{BA}=$ Write$_{AB}$ 
(see Definition \ref{def1}).
\end{thm}
\begin{proof} Process $B$ is the $i\/$th neighbour of $A$. 
Between each of its two writings, $A$ enters the {\bf repeat} 
loop and exits the loop. Once $A$ is within the loop, the register 
$Write_{AB}$ does not change value. The {\bf repeat} loop's 
code is such that when the loop is exited, the value of the local 
variable $s_i$ of $A$ and the value of the register $Write_{AB}$ 
are equal. In the loop, the local variable $r_i$ of $A$ takes the 
value of the register $Read_{AB}$. The value of the register 
$Read_{BA}$ may change after this assignment and before the loop 
is exited. Thus, when the loop is exited two distinct cases have 
to be considered:

$\bullet$\ \, No update of the register $Read_{BA}$ happens 
between the reading from that register and the loop exit. Then, 
$s_i=$ {\em Write}$_{AB}=val_i=$ {\em Read}$_{BA}$, and 
$B$ allows the writing of $A$.

$\bullet$\ \, Writings into the register $Read_{BA}$ happen 
between the reading from that register and the loop exit. 
However, the latter writings are performed within the context 
of updating. Hence, each time the value has changed, 
we have that {\em Read}$_{BA}=$ {\em Write}$_{AB}$ and, by 
Lemma~\ref{stop}, the equality holds while $A$ does not rewrite 
into the register {\em Write}$_{AB}$.
\end{proof}

After the writing of a value in the register $Write_{AB}$, 
the first primitive guarantees that $A$ will only write 
in the register $Write_{AB}$ if $B$ allows it. In the case 
when the value is new, $B$ must perform the action 
{\bf read($Write_{AB}$)} to allow the writing.

\subsection*{Summing up of the Results}
\begin{enumerate}
\item  {\bf The protocol is live:} every process is 
updating all its {\em Write} registers an infinite number 
of times.
\item  {\bf The protocol is correct:} no process can write 
distinct values twice in a row in its {\em Write} register 
without any previous reading from that register.
\end{enumerate}

\section {The Weak Rendezvous Protocol} \label{wrv}
In this section, we present a self-stabilizing {\em weak rendezvous} 
communication primitive.

Recall that The {\em rendezvous} mechanism (as defined in~\cite{hoa}) 
synchronizes communication in the link-register asynchronous 
model of distributed system: each $write$ or $read$ operation 
is performed in and from the same register. When Process $A$ writes 
a value in its register $Write_{AB}$, it cannot perform 
{\em any other action} until process $B$ has completed a 
$read$ operation from the register $Write_{AB}$.

\medskip
The {\em weak rendezvous} mechanism only requires that between 
two $write$ operations performed by a process $A$ in $Write_{AB}$, 
process $B$ performs at least one $read$ operation from $Write_{AB}$. 
Therefore, if $A$ writes a value $n$ consecutive times 
(even the same ones in each row) in the register $Write_{AB}$, 
the primitive guarantees that $B$ copies each of the $n$ values 
at least one time, once the system is stabilized.

The weak rendezvous mechanism is based upon the alternating bit
technique. After writing in its register $Write_{AB}$, process 
$A$ changes the value of the bit-register $Control_{AB}$. 
$A$ can write again in the register $Write_{AB}$ only after 
$B$ has copied the new value of $Control_{AB}$ into the 
register $CheckControl_{BA}$. And $B$ copies the value only 
after reading in the register $Write_{AB}$.

The liveness proof of the weak rendezvous protocol is similar 
to the proof of the read checking protocol.
The following Theorem~\ref{proof2} proves the correctness 
of the weak rendezvous protocol.
 
\begin{thm} \label{proof2}
Let A and B be two neighbouring processes. After B executes 
its first action and after the xth ($\ge 2$) writing in the 
register Write$_{AB}$, B reads the value from Write$_{AB}$ 
before the next writing in Write$_{AB}$.
\end{thm}
\begin{proof}
As shown in Theorem~\ref{proof}, we can establish
that before the $x\/$th writing in the register $Write_{AB}$,
$Control_{AB} = CheckControl_{BA}$. After the writing in 
the register $Write_{AB}$, $A$ changes the value in $Control_{AB}$ 
and enters the {\bf repeat} loop ($Control_{AB}\ne CheckControl_{BA}$).
$A$ stays within the loop as long as $B$ does not copy  
the value of $Control_{AB}$ into the register $CheckControl_{BA}$. 
Finally, $B$ copies the value only after reading in the register 
$Write_{AB}$.
\end{proof}

The weak rendezvous protocol maintains a weak scheduling of the 
communication between processes in the following sense. 
We call a {\em weak scheduling} of the communication between 
process $A$ and all its $N_A$ neighbours the property that 
$A$ can write twice into its registers {\em Write}$_{AB_i}$, 
only whenever all the $B_i$'s did read from the register 
{\em Write}$_{AB_i}$ in the meantime $(1\le i\le N_A)$. 

%\newpage
\noindent \hrulefill\-\hrulefill

\smallskip
\noindent {\bf constant}\ \, $N_A $\hspace{.4cm} : the number of 
neighbours of $A$ ; \\
\noindent {\bf var}\ \,\, $r_i$\hspace{1.5cm} : message sent from 
the $i\/$th neighbour of $A$ ; \\
\hspace*{.8cm} $b_i $\hspace{1.5cm} : alternate bit sent from the 
$i\/$th neighbour of $A$ ; \\
\hspace*{.8cm} $c_i $\hspace{1.5cm} : alternate bit sent from A to 
the $i\/$th neighbour of $A$ ; \\
\hspace*{.8cm} $l_i$\hspace{1.5cm} : value of the last alternate bit
sent from $A$ and read by the $i\/$th neighbour of $A$;\\

\newpage
%\medskip
\noindent {\bf while} {\em true} {\bf do} \\
\hspace*{.5cm} {\bf for} $\;i=1\;$ {\bf to} $\;N_A\;$ {\bf do} \\
\hspace*{2cm} {\bf write}$(Write_{AB_i},get_i)$ ; \\
\hspace*{2cm} $c_i \gets$ {\bf read}$(Control_{AB_i})$ ; \\
\hspace*{2cm} {\bf write}$(Control_{AB_i},(c_i+1)\bmod 2)$ ; \\
\hspace*{.5cm} {\bf endfor} \\
\hspace*{.5cm} {\bf repeat} \\
\hspace*{1cm} {\bf for} $\;i=1\;$ {\bf to} $\;N_A\;$ {\bf do} \\
\hspace*{2cm} $r_i\gets$ {\bf read}$(Write_{B_iA})$ ; \\
\hspace*{2cm} $b_i\gets$ {\bf read}$(Control_{B_iA})$ ; \\
\hspace*{2cm} {\bf write}$(CheckControl_{AB_i},b_i)$ ; \\
\hspace*{2cm} $c_i\gets$ {\bf read}$(Control_{AB_i})$ ; \\
\hspace*{2cm} $l_i\gets$ {\bf read}$(CheckControl_{B_iA})$ ; \\
\hspace*{1cm} {\bf endfor} \\
\hspace*{.5cm} {\bf until}\ \, \,( $\forall i\in [1,N_A]\ \; c_i = l_i$ ) \\
{\bf endwhile}
\smallskip
\begin{center}
{\bf Fig.~3.}\ {\em The weak rendezvous protocol for $A$.}
\end{center}
\vspace{-.3cm}
\noindent \hrulefill\-\hrulefill

\section {The Quasi Rendezvous Protocol} \label{qrv} 
In this section, we present a self-stabilizing {\em quasi rendezvous} 
communication primitive. A close idea may be found in~\cite{anh}, 
where the authors also present a self-stabilizing data link protocol 
under read/write atomicity such that, between two $write$ 
operations in the register, there is only one $read$ operation from 
that register. (See our remarks in section~\ref{related}.)

\medskip 
The {\em quasi rendezvous} mechanism requires that between two 
$write$ operations performed by the process $A$ in $Write_{AB}$, 
the process $B$ cannot perform but one and only one {\em read} 
operation from $Write_{AB}$. Therefore, if $A$ writes $n$ 
consecutive times the same value (possibly the same one in each row) 
in the register $Write_{AB}$, the primitive guarantees that $B$ 
will copie each of the $n$ values {\em exactly one time}, once the 
system is stabilized. 
 
The quasi rendezvous mechanism is based upon the alternating bit 
technique. After reading from the register $Write_{AB}$, the 
process $B$ copies the value of the bit-register $Control_{AB}$ 
into $CheckControl_{BA}$. Now, $B$ can read again from the register 
$Write_{AB}$ only after $A$ has changed the value of $Control_{AB}$. 
And $A$ changes that value only after writing in the register 
$Write_{AB}$. 

\medskip
\noindent \hrulefill\-\hrulefill 
 
\smallskip 
\noindent {\bf constant}\ \, $N_A $\hspace{.4cm} : the number of 
neighbours of $A$ ; \\ 
\noindent {\bf var}\ \,\, $r_i$\hspace{1.5cm} : message sent from 
the $i\/$th neighbour of $A$ ; \\ 
\hspace*{.8cm} $b_i $\hspace{1.5cm} : alternate bit sent from the 
$i\/$th neighbour of $A$ ; \\ 
\hspace*{.8cm} $c_i $\hspace{1.5cm} : alternate bit sent from A to 
the $i\/$th neighbour of $A$ ; \\ 
\hspace*{.8cm} $l_i$\hspace{1.5cm} : value of the last alternate bit 
sent from $A$ and read by the $i\/$th neighbour of $A$;\\ 
\hspace*{.8cm} $d_i$\hspace{1.5cm} : value of the last alternate bit 
sent from the $i\/$th neighbour of $A$  and read by $A$\\ 
 
\medskip 
\noindent {\bf while} {\em true} {\bf do} \\ 
\hspace*{.5cm} {\bf for} $\;i=1\;$ {\bf to} $\;N_A\;$ {\bf do} \\ 
\hspace*{2cm} {\bf write}$(Write_{AB_i},get_i)$ ; \\ 
\hspace*{2cm} $c_i \gets$ {\bf read}$(Control_{AB_i})$ ; \\ 
\hspace*{2cm} {\bf write}$(Control_{AB_i},(c_i+1)\bmod 2)$ ; \\ 
\hspace*{.5cm} {\bf endfor} \\ 
\hspace*{.5cm} {\bf repeat} \\ 
\hspace*{1cm} {\bf for} $\;i=1\;$ {\bf to} $\;N_A\;$ {\bf do} \\ 
\hspace*{2cm} $b_i\gets$ {\bf read}$(Control_{B_iA})$ ; \\ 
\hspace*{2cm} $d_i\gets$ {\bf read}$(CheckControl_{AB_i})$ ; \\ 
\hspace*{2cm} {\bf if} $b_i \neq d_i$ {\bf then} \\ 
\hspace*{3cm} $r_i\gets$ {\bf read}$(Write_{B_iA})$ ; \\ 
\hspace*{3cm} {\bf write}$(CheckControl_{AB_i},b_i)$ ; \\ 
\hspace*{2cm} {\bf endif} \\ 
\hspace*{2cm} $c_i\gets$ {\bf read}$(Control_{AB_i})$ ; \\ 
\hspace*{2cm} $l_i\gets$ {\bf read}$(CheckControl_{B_iA})$ ; \\ 
\hspace*{1cm} {\bf endfor} \\ 
\hspace*{.5cm} {\bf until}\ \, \,( $\forall i\in [1,N_A]\ \; c_i = l_i$ ) \\ 
{\bf endwhile} 
\smallskip 
\begin{center} 
{\bf Fig.~4-.}\ {\em The quasi rendezvous protocol for $A$.} 
\end{center} 
\vspace{-.3cm} 
\noindent \hrulefill\-\hrulefill 

\medskip 
The liveness proof of the quasi rendezvous protocol is similar 
to the proof of the read checking protocol. 
 
\begin{defi} \label{def2} 
Let A and B be two neighbouring processes. B is said to allow A 
to write iff CheckControl$_{BA}=$ Control$_{AB}$. 
\end{defi} 
 
\begin{defi} 
Let A and B be two neighbouring processes. The full reading 
of register Write$_{AB}$ is completed by the sequence of the 
four following actions performed by B: 
\newline 
$b\gets$ {\bf read}$(Control_{BA})$~; 
$d\gets$ {\bf read}$(CheckControl_{AB})$~; 
{\bf if} $b \neq d$ {\bf then} 
$\{r\gets$ {\bf read}$(Write_{BA})$~; 
{\bf write}$(CheckControl_{AB},b)~;\}$. 
\end{defi} 
 
\begin{defi} 
Let A and B be two neighbouring processes. The full writing 
of register Write$_{AB}$ is completed the sequence of the three 
following actions performed by A: 
\newline 
{\bf write}$(Write_{AB},get)$~; 
$c \gets$ {\bf read}$(Control_{AB})$~; 
{\bf write}$(Control_{AB},(c+1)\bmod 2)$~; 
\end{defi} 
 
\begin{lem} \label{allow2} 
Let A be a process with its program counter in the {\bf repeat} loop 
and let B be a neighbour of A. Whatever the current configuration 
and the execution, the system of processes executing the protocol 
either eventually reaches a configuration in which B allows A to 
write, or A exits the {\bf repeat} loop. 
\end{lem} 
 
\begin{lem} \label{wrong2} 
After executing its first three actions, no process can perform an 
incomplete reading or writing. 
\end{lem} 
 
\begin{lem} \label{stop2} 
Let A and B be two neighbouring processes. After B and A execute 
their first three actions, if B allows A to write, then only 
the complete writing of A in its register Write$_{AB}$ may be able 
to cancel that permission. 
\end{lem} 
\begin{proof}
The proof of the three above lemmas (\ref{allow2}, \ref{wrong2} and 
\ref{stop2}) is similar to the proof of Lemma~\ref{allow}, 
Lemma~\ref{wrong} and Lemma~\ref{stop}, respectively.
\end{proof}

\begin{thm} \label{exit2} 
Let A be a process. Whatever the execution, the system of processes 
which performs the protocol reaches a configuration in which 
A is not within the {\em \bf repeat} loop anymore. 
\end{thm} 
\begin{sketchproof}
The proof is by contradiction and it is similar to the proof of 
theorem~\ref{exit}. 
\end{sketchproof}

\begin{cor} 
Let A be a process. Whatever the execution, A writes an infinite 
number of times into all its Write registers. 
\end{cor} 
The following Theorems \ref{proof3.a} and \ref{proof3.b} prove the 
correctness of the quasi rendezvous protocol. 
 
\begin{thm} \label{proof3.a} 
Let A and B be two neighbouring processes. After A and B execute 
their first  three actions and after the xth ($\ge 2$) writing 
in the register Write$_{AB}$, B reads the value from Write$_{AB}$ 
before the next writing in Write$_{AB}$ can take place. 
\end{thm} 
\begin{proof} 
We can establish that before the $x\/$th writing in the register 
$Write_{AB}$, $Control_{AB}=CheckControl_{BA}$. After writing into 
the register $Write_{AB}$, $A$ changes the value in $Control_{AB}$ 
and enters the {\bf repeat} loop ($Control_{AB}\ne CheckControl_{BA}$). 
$A$ stays within the loop as long as $B$ does not copy 
the value of $Control_{AB}$ into the register $CheckControl_{BA}$. 
Finally, $B$ copies the value only after reading from the register 
$Write_{AB}$. 
\end{proof} 
 
\begin{thm} \label{proof3.b} 
Let A and B be two neighbouring processes. After A and B execute 
their first three actions and after B reads from Write$_{AB}$, 
A performs a complete writing in Write$_{AB}$ before the next reading 
from Write$_{AB}$. 
\end{thm} 
\begin{proof} 
Before the reading from $Write_{AB}$, $Control_{AB}\ne CheckControl_{BA}$. 
After the reading from the register $Write_{AB}$, $B$ changes the value 
in $CheckControl_{BA}$ Now, $B$ does not change the value in 
$CheckControl_{BA}$ ($B$ does not read from the register $Write_{AB}$) 
as long as $A$ does not change the value in $Control_{AB}$. After the 
first three actions of $A$, changing the value in $Control_{AB}$ is made 
after $A$'s writing in $Write_{AB}$. 
\end{proof} 
 
The quasi rendezvous protocol maintains a scheduling of the communications 
between processes in the following sense. We call a {\em scheduling} 
of communications between process $A$ and all its $N_A$ neighbours 
the property that $A$ can write twice into its registers 
{\em Write}$_{AB_i}$, only whenever each of the $B_i$'s performed 
{\em one unique} reading from the register {\em Write}$_{AB_i}$ 
in the meantime $(1\le i\le N_A)$. 

\section{Concluding Remarks}
The paper presents three very basic general protocols for the design 
of fair and reliable self-stabilizing communication primitives. 
Both protocols work in arbitrary networks and also ensure minimal 
scheduling properties, whatever the initial configuration of the 
system of processes and the activations by the scheduler. In particular, 
the last protocol entails the mechanism of a ``quasi rendezvous'', 
which proves useful in more involved self-stabilizing protocols.

Each primitive can actually be used as a ``black box'' by a separate 
protocol, handling the procedures in more involved self-stabilizing 
algorithms. Thus, the protocols may be modified according to the
designer's will and needs: e.g., in specific topologies of networks 
a weak scheduling of communications may impose fewer neighbours 
to read from the registers. For example, with only one neighbour, 
a point to point self-stabilizing quasi rendezvous mechanism may 
also be completed. Along the same lines, the protocols also simulate 
reliable self-stabilizing message-passing in asynchronous distributed 
systems.

Although the paper does not concern itself with complexity measures, 
it is worth mentioning that when time is measured by some appropriately 
defined round complexity, the stabilization time of the read checking 
protocol is $O(1)$.

\end{document}